\documentclass[letterpaper, 12pt]{article}

\title{On the Multidimensional Stable Marriage Problem}
\author{Jared Duker Lichtman \\ \small{Dartmouth College} \\ \footnotesize{\texttt{lichtman.18@dartmouth.edu}}}
\date{} %leave blank for now

\usepackage[paperwidth = 8.5in, paperheight = 11in, margin = 1in]{geometry}
\usepackage{tikz}
\usepackage{multirow}
\usepackage{amsthm}
\usepackage{amsmath}
\usepackage{amsfonts}
\usepackage{upgreek}

\begin{document}

%Use this to find unknown mathematics symbols: 
%"detexify.kirelabs.org/classify.html"

\maketitle

\abstract{We provide a problem definition of the stable marriage problem for a general number of parties $p$ under a natural preference scheme in which each person has simple lists for the other parties. We extend the notion of stability in a natural way and present so called elemental and compound algorithms to generate matchings for a problem instance. We demonstrate the stability of matchings generated by both algorithms, as well as show that the former runs in $O(pn^2)$ time.}

\section{Introduction}

\hspace*{\parindent} The stable marriage problem (SM) is a famous problem in mathematics in which there exists a community of $n$ men and $n$ women, all of whom are to be paired with each other heterogeneously in marriage.  Each individual provides a complete preference list ranking the members of the opposite sex according to his or her preference for marriage. The final matching sought is one in which no two people would rather be married to each other over their current spouses, and is thus called \textbf{stable}. It is of note that stability is a heuristic approach to optimizing matching.

In 1962, David Gale and Lloyd Shapley \cite{gs} presented an algorithm to solve SM. In the Gale-Shapley (GS) algorithm, each man proposes to his favorite woman, and each woman is temporarily matched to the man that proposes to her.  If a woman is proposed to by multiple men, she is temporarily matched to the proposing man she prefers the most. If there remain any unmatched people after the first round of matching, the single men propose again, this time to their second choices.  If a woman who has already been matched is proposed to by a man whom she prefers over her current partner, she will leave her current partner and become paired with the new man.  This process is repeated, with each unmatched man proposing to the women in order of his preference list and each woman choosing her best possible mate who has proposed to her, until everyone has been matched.

In \cite{gs}, the following properties of the GS algorithm were shown:
\begin{itemize}
\item It terminates
\item It is $O(n^2)$; maximal number of rounds is $n^2 - 2n + 2$
\item Resultant matching is stable
\item matching is optimal (of stable) for proposing party
\item matching is pessimal (of stable) for responding party
\end{itemize}

For a thorough introduction to the stable marriage problem, including the Gale-Shapley algorithm, we refer the reader to \cite{gusfieldirving}. 

In 1976, Donald Knuth proposed twelve open questions \cite{knuth} on SM, one of which asked to generalize SM from two to three parties---the 3-dimensional SM (3DSM). Given the open-ended wording, the question has been addressed under various interpretations with respect to structure of preferences and definitions of stability. With the addition of indifference in preferences \cite{irving}, there have been many different constructions of 3DSM, most of which have been shown to be $\mathcal{NP}$-complete \cite{nghirsch}\cite{huang}.

However, 3DSM has been shown to work under a simple scheme \cite{danilov} in which each individual provides two simple preference lists for the other two parties. Along this vein, this paper will consider the $p$-dimensional SM ($p$DSM), where $p \geq 2$, and propose two types of algorithms to deal with it. Before presenting the algorithms in their entirety, we motivate them with the simplest nontrivial consideration, when $p=3$.

In an instance of 3DSM, we have a community of men, women, and dogs. One possible way to match the men, women, and dogs together is to choose two parties, say men and women, and create a matching between them using the Gale-Shapley algorithm.

Repeating this with another pair of parties, say women and dogs, we arrive at sets of man-woman and woman-dog pairs, from which we can deduce a matching for the entire community. This is the idea motivating the \textbf{elemental algorithm}.

Alternatively, we can take the man-woman pairs and view each as members of a compound party, the humans. Each human's preferences are constructed by combining each (associated) man's and woman's preferences for dogs. Similarly, we can modify each dog's preferences by combining its preference of men and women into preferences for humans. This is the idea motivating the \textbf{compound algorithm}.

Here the names elemental and compound are chosen in allusion to chemistry; elements are the fundamental building blocks, and compounds are created by bonding elements together. In the elemental algorithm, each party is treated as an individual element in ``pure" form, whereas in the compound algorithm, parties are ``bonded" together via stable matchings.

\section{Problem Definition}

\hspace*{\parindent} We proceed with a formal construction of the multidimensional stable marriage problem. An instance of the $p$-dimensional stable marriage problem ($p$DSM) is an ordered pair $(\mathcal{P}, L)$ with a set $\mathcal{P}$ of $p$ disjoint parties with $n$ elements each, and a set $L$ of preference lists for every element (to be defined below). An element of a given party is said to be a \textbf{member} of that party. Let $U_\mathcal{P} =\bigcup_{P \in \mathcal{P} } P$ be the \textbf{community} of $\mathcal{P}$. For all $x \in U_\mathcal{P}$, let $prt(x)$ return the party of which $x$ is a member.

Associated with each $x \in U_\mathcal{P}$ is a $p-1 \times n$ \textbf{strictly ordered preference array} $L_x$  of $x$'s preferences defined as follows. For all $y \in P \in \mathcal{P}$ where $x \notin P$, let $L_x(y)=j$ denote that $y$ is $x$'s $j^{th}$ preferred member of $P$. Let 
$$L_x(P)=\{ y \in P\: |\: \text{strictly ordered by } L_x(y) \}$$
$$L_x=\{ L_x(P) \: | \: P \in \mathcal{P} \}$$
$$L=\{ L_x \: | \: x \in U_\mathcal{P} \}$$
Furthermore, given $a,b \in P$, let $a \succ_x b$ denote that $L_x(a) < L_x(b)$ and let $a \succeq_x b$ denote that $L_x(a) \leq L_x(b)$.

A \textbf{family} $F \subseteq U_\mathcal{P}$ is a set of $p$ elements, one member from each party. A \textbf{matching} $\mathcal{F}$ over $\mathcal{P}$ is a partition of $U_\mathcal{P}$ into $n$ families. Elements of a single family are said to be \textbf{relatives} in $\mathcal{F}$. Let $rel_\mathcal{F}(x, P)$ return $x$'s relative in $\mathcal{F}$ from $P$. 

For a given $\mathcal{F}$, a family $F \notin \mathcal{F}$ is \textbf{blocking} if and only if 
\begin{itemize}
\item $x \succeq_y rel_\mathcal{F}\big( y, prt(x)\big) \; , \; y \succeq_x rel_\mathcal{F}\big( x, prt(y)\big) \quad \text{ for all } x,y \in F$
\item for each $x$, there exists $z \in F$ such that $z \succ_x rel_\mathcal{F}\big( x, prt(z)\big)$
\end{itemize}
A matching $\mathcal{F}$ is \textbf{unstable} if there exists a blocking family in $U_\mathcal{P}$. $\mathcal{F}$ is otherwise \textbf{stable}.

There are several possible ways to define the problem, with regards to both manner of describing preferences and definition of stability. The manner of preferences was chosen, in part, because of its simplicity. In 3DSM, simple preference lists leads to an efficient algorithm, while more permutation-oriented (Cartesian product) setups have generated complex problems proven to be $\mathcal{NP}$-complete. The definition of stability chosen resembles that of the traditional problem most closely. This is because if there is a blocking family of elements who all prefer each other to their corresponding partners in $\mathcal{F}$, they would all ``elope" and desert their established families in real life. Whereas setups that permit indifference among individuals have led to $\mathcal{NP}$-complete problems as well \cite{irving}.

Two types of algorithms will be presented for $p$DSM, elemental and compound, both of which are novel extensions of the original GS algorithm. Given any algorithm $A$, let $\langle A\rangle$ denote the matching generated by $A$. Given a set of algorithms $\mathcal{A}$, let $\langle\mathcal{A}\rangle = \{ \langle A \rangle \: | \: A \in \mathcal{A} \}$.

\section{Elemental Algorithms}
A tentative definition of an elemental algorithm will initially be provided, which will prompt a deeper understanding and a more rigorous definition.
Let $GS(P, Q)$ denote that $P$ proposes to $Q$ according to the GS algorithm. An important way of viewing $GS(P, Q)$ is that it establishes a bijection between $P$ and $Q$. 

An \textbf{elemental algorithm} is a set $\upvarepsilon$ of bijections $GS(R, S)$ such that for all $P, Q \in \mathcal{P}$ there exists a unique bijection (either directly, or indirectly by composition) between $P$ and $Q$. To execute $\upvarepsilon$, each element in $\upvarepsilon$ is executed, generating a unique 1-1 correspondence between each pair of parties, and thus a matching for the problem.

In SM, the GS algorithm is executed on the pair of men and women. This may be viewed as a directed graph with a vertex for each party and an edge directed from the proposing into the responding vertex. Unless otherwise specified, all graphs are assumed to be simple and labeled.

In $p$DSM, these graphs may be similarly constructed. Given an elemental algorithm $\upvarepsilon$, an \textbf{elemental graph} $G=\{ V,E\}$ is generated according to the following bijections.
$$\mathcal{P} \rightarrow V : P \mapsto v_p \qquad \quad \upvarepsilon \rightarrow E : GS(P,Q) \mapsto e_{pq}$$

The following theorem will provide a basis for the formal definition of an elemental algorithm.

\newtheorem{thm}{Theorem}
\begin{thm}
$G$ is an elemental graph $\Longleftrightarrow G$ is a tree.
\end{thm}
\begin{proof}
This bidirectional proof will be split into two parts. The first will be proven by contradiction, and the second by direct construction.

\textbf{part 1:} Take an elemental graph $G$. Since $e_{pq}$ represents a bijection between $P$ and $Q$, an (undirected) walk $W$ in $G$ establishes a composition of bijections, relating all parties with corresponding vertices in $W$ to each other. Since a bijection is established between all pairs of parties, $G$ must be connected. Further, since each bijection is unique, there can be no cycles in $G$. So $G$ is a tree.

\textbf{part 2:} Take a tree $T=\{ V,E\}$. Let $T$ generate a set $\updelta$ of bijections $GS(P, Q)$ according to the following bijections:
$$V \rightarrow \mathcal{P} : v_p \mapsto P \qquad \quad E \rightarrow \updelta : e_{pq} \mapsto GS(P,Q)$$
Since $T$ is a tree, there exists a unique walk between $v_p$ and $v_q$, corresponding to a unique bijection between $P$ and $Q$. Thus, by definition, $\updelta$ is an elemental algorithm. Since the maps used to generate $\updelta$ from $T$ were bijective, $\updelta$ may be thought to generate $T$ using the inverse maps. Then, since $\updelta$ is an elemental algorithm, $T$ is an elemental graph.
\end{proof}
The formal definition follows from the reasoning in part 2 of Theorem 1. Given a directed tree $T=\{ V,E\}$ with $p$ vertices, an \textbf{elemental algorithm} $\upvarepsilon_T$ over $\mathcal{P}$ is a set of bijections $GS(P, Q)$ generated according to the following bijections:
$$V \rightarrow \mathcal{P} : v_p \mapsto P \qquad \quad E \rightarrow \upvarepsilon_T : e_{pq} \mapsto GS(P,Q)$$
The execution of $\upvarepsilon_T$ remains unchanged; each element of $\upvarepsilon_T(\mathcal{P})$ is executed, generating a matching $\langle\upvarepsilon_T(\mathcal{P})\rangle$. 

It is important to note here that a \textit{directed} tree is specified. The direction of the edge indicates which party proposes, an important fact when executing the algorithm. However, such specification was unnecessary in the proof of Theorem 1 because both are viable elemental algorithms.

Having provided a definition, we now show that the elemental algorithm shares many of the nice properties the GS algorithm has. Let $\mathcal{T}_p$ be the set of all undirected trees with $p$ vertices. Let $\mathcal{E}(P)=\{\upvarepsilon_T(\mathcal{P}) \: |\: T \in \mathcal{T}_p \}$. For ease (given appropriate context), let $\upvarepsilon$, $\mathcal{T}$ and $\mathcal{E}$ be short for $\upvarepsilon_T$, $\mathcal{T}_p$, and $\mathcal{E}(\mathcal{P})$, respectively.

\begin{thm}
All elemental algorithms terminate.
\end{thm}
\begin{proof}
Take $\upvarepsilon \in\mathcal{E}$. By construction, $\upvarepsilon$ contains $p-1$ instances of the GS algorithm. Since the GS algorithm terminates \cite{gs} and $p-1$ is finite, $\upvarepsilon$ must terminate.
\end{proof}

\begin{thm}
All elemental algorithms yield stable matchings.
\end{thm}
\begin{proof}
Consider an arbitrary $\upvarepsilon \in \mathcal{E}$. Suppose $\langle\upvarepsilon\rangle$ is unstable; there exists a blocking family $F \notin \langle\upvarepsilon\rangle$. Take $x,y \in F$ where $x \in P, y \in Q$.

\textbf{case 1}: $GS(P, Q) \in \upvarepsilon$. Since $GS(P, Q)$ is stable \cite{gs}, we have $x = rel_{\langle\upvarepsilon\rangle}(y, P)$.

\textbf{case 2}: $GS(P, Q) \notin \upvarepsilon$. Take the tree $T=\{ V,E \}$ that generated $\upvarepsilon$. There is a unique walk $W$ in $T$ from $v_p$ to $v_q$. By relabeling vertices, without loss of generality, let
\begin{center}
$W=u_1 u_2 u_3 \cdots u_w$ \quad for some $1 < w < p$, where $u_1=v_p$ and $u_w=v_q$
\end{center}
Relabeling parties accordingly, we have that
\begin{center}
$GS(P_k, P_{k+1}) \in \upvarepsilon$ \quad for $1 \leq k < w$, where $P_1=P$ and $P_w=Q$
\end{center}
We now apply case 1 $w-1$ times on $x_k,x_{k+1}$, where $x_k \in P_k, x_1=x$ and $x_w=y$. This gives $x_k=rel_{\langle\upvarepsilon\rangle}(x_{k+1}, P_k)$. Then by transitivity, $x_1=rel_{\langle\upvarepsilon\rangle}(x_w, P_1)$, or $x=rel_{\langle\upvarepsilon\rangle}(y, P)$.

Combining both cases, we have that $F \in \langle\upvarepsilon\rangle$, which is a contradiction. Therefore $\langle\upvarepsilon\rangle$ is stable.
\end{proof}

For 3DSM, this setup was criticized for its simplicity: ``It is not hard to see that we can apply the Gale-Shapley algorithm twice to get a weak stable matching: letting the men propose to women and then propose to dogs. Women and dogs make the decision of acceptance or rejection based on their simple lists of men." \cite{huang} However, previous literature has overlooked the multitude of combinations of matchings available, allowing one to customize the algorithm to fit the given task at hand. The following result computes the exact number of such options.

\begin{thm}
$|\mathcal{E}|=2^{p-1}p^{p-2}$
\end{thm}
\begin{proof}
Since each $T \in \mathcal{T}$ generates a distinct $\upvarepsilon (T) \in \mathcal{E}$, we have that $|\mathcal{T}|=|\mathcal{E}|$. By Cayley's Formula, there are $p^{p-2}$ undirected trees with $p$ vertices. Since there are $p-1$ edges to a tree, each undirected tree corresponds to $2^{p-1}$ distinct directed trees. Thus $|\mathcal{E}|=|\mathcal{T}|=2^{p-1}p^{p-2}$.
\end{proof}

\begin{thm}
An elemental algorithm is $O(pn^2)$.
\end{thm}
\begin{proof}
Recall that the GS algorithm is $O(n^2)$ and the maximal number of rounds is $n^2 - 2n + 2$. Since an elemental algorithm applies the GS algorithm $p-1$ times, the maximal number of rounds in an elemental algorithm is $(p-1)(n^2 - 2n + 2)$. Therefore, it is $O(pn^2)$.
\end{proof}

\section{Problem Structure}
Before defining compound algorithms formally, we need to develop a vocabulary for some of the natural structure of a multidimensional stable marriage instance.

Given a $p$DSM $\mathcal{P}$, a $q$DSM $\mathcal{Q}$ is said to be a \textbf{subproblem} of $\mathcal{P}$ if and only if $\mathcal{Q} \subseteq \mathcal{P}$. Given $\mathcal{P}$, a set partition $\pi(\mathcal{P})\neq \mathcal{P}, \{\mathcal{P}\}$ is said to be a \textbf{problem partition} containing subproblems of $\mathcal{P}$. Let $\Pi(\mathcal{P})$ be the set of all problem partitions. For ease, let $\pi$, $\Pi$ be short for $\pi(\mathcal{P})$, $\Pi(\mathcal{P})$, respectively.

Given a problem partition $\pi(\mathcal{P})$ containing $1\leq p' < p$ subproblems with a matching $\mathcal{F}$ over each $\mathcal{Q} \in \pi$, then a \textbf{reduced problem} $(\mathcal{P'}, L')$ is a $p'$DSM where $\mathcal{P'}=\pi$ and $L'$ is defined as follows. Given matchings $\mathcal{F}, \mathcal{G}$ for subproblems $\mathcal{Q}, \mathcal{R} \in \mathcal{P'}$, respectively, for all families $F \in \mathcal{F}, G \in \mathcal{G}$, let
$$L'_F(G)=\sum\limits_{x \in F, \: y \in G} L_x(y)$$
$$L'_F(\mathcal{R})=\{G \in \mathcal{G} \: | \: \text{strictly ordered according to } L'_F(G) \}$$
$$L'_F = \{ L'_F(\mathcal{R}) \: | \: \mathcal{R} \in \mathcal{P'}\}$$
$$L' = \{ L'_F \: | \: F \in U_\mathcal{P'} \}$$
The definition of the reduced problem effectively collapses each subproblem $\mathcal{Q} \in \mathcal{P'}$ into a single party and each family $F \in \mathcal{F}$ into a single individual.

%% discuss the consequences of ties in discussion

Given a matching $\mathcal{F'}$ over reduced problem $(\mathcal{P'}, L')$, $\mathcal{F'}$ may be \textbf{expanded} to $\mathcal{F}$ according to 
$$\mathcal{F}=\bigcup_{F'\in\mathcal{F'}} F'$$
to give the equivalent matching $\mathcal{F}$ over the original problem $(\mathcal{P}, L)$.

\section{Compound Algorithms}

Given a $p$DSM $(\mathcal{P}, L)$, a \textbf{compound algorithm} $C$ is executed over $\mathcal{P}$ according to the following recursive procedure. Two counters $i, c$ are both intially set to zero.
\begin{enumerate}
\item Take an elemental algorithm $\upvarepsilon \in \mathcal{E}(\mathcal{P})$. If $\Pi(\mathcal{P})\neq\emptyset$, take a problem partition $\pi \in \Pi(\mathcal{P})$. Else, let $\langle C\rangle=\langle\upvarepsilon\rangle$ and go to step 6.
\item Create reduced problem $(\mathcal{P'}, L')$: set $\mathcal{P'}=\pi$; for each subproblem $\mathcal{Q} \in \pi$, take an elemental algorithm $\upvarepsilon\in\mathcal{E}(\mathcal{Q})$ and generate $L'$ using matchings $\langle\upvarepsilon(\mathcal{Q})\rangle$.
\item Index $i$ by one.
\item If $|\pi|>1$, repeat from step 1 letting $(\mathcal{P}, L)=(\mathcal{P'}, L')$. Else, let $c=i$.
\item Given $\langle\upvarepsilon(\mathcal{P'}) \rangle$ used in the $c^{th}$ execution of step 2, construct $\langle C\rangle$ as the result of expanding $\langle\upvarepsilon(\mathcal{P'}) \rangle \: c$ times.
\item Return $\langle C\rangle$.
\end{enumerate}

Having provided a definition, we now show that compound algorithms share many of the nice properties the elemental algorithms have. Let $\mathcal{C}(\mathcal{P})$ be the set of all compound algorithms over $\mathcal{P}$.
\begin{thm}
All compound algorithms terminate.
\end{thm}
\begin{proof}
The statement will be proven by induction on $|\mathcal{P}|$. Take $C \in \mathcal{C}(\mathcal{P})$.

\textbf{Base case:} $|\mathcal{P}|=2$

Let $\mathcal{P}=\{P,Q\}$. Without loss of generality, $\upvarepsilon=GS(P,Q)$. The only possible set partitions are $\mathcal{P}, \: \{\mathcal{P}\}$, neither of which are problem partitions. Thus $\Pi(\mathcal{P})=\emptyset$; for the execution of $C$, we go from step 1 to step 6, return $\langle GS(P,Q)\rangle$, and terminate.

\textbf{Inductive step:} Given $|\mathcal{P}|=k+1$ and $C(\mathcal{Q})$ terminates when $|\mathcal{Q}|=i$ for all $2 \leq i \leq k$.

Recall that a problem partition omits the trivial partition $\pi=\{\mathcal{P}\}$. This ensures, for all $\pi \in \Pi(\mathcal{P})$, that $|\mathcal{P}| > |\mathcal{Q}|$ for all $\mathcal{Q} \in \pi$. Thus $k+1 > |\mathcal{Q}|$, meaning that $|\mathcal{Q}|=i$ for some $2 \leq i \leq k$. By the inductive hypothesis, $C(\mathcal{Q})$ terminates. Therefore, so does $C(\mathcal{P})$.
\end{proof}

Before showing stability, we need a lemma.

\newtheorem{lma}{Lemma}
\begin{lma}
Given $p$DSM $(\mathcal{P}, L)$ reduced to $(\mathcal{P'}, L')$ using stable matchings $\mathcal{F}_\mathcal{Q}$ for all $\mathcal{Q} \in \mathcal{P'}$. If matching $\mathcal{F'}$ over $\mathcal{P'}$ is stable, then the expanded matching $\mathcal{F}$ over $\mathcal{P}$ is stable.
\end{lma}
\begin{proof}
The proof will be given by taking any $\mathcal{Q}, \mathcal{R} \in \mathcal{P'}$ and showing that the expansion of $\mathcal{F'} \cap (\mathcal{Q} \cup \mathcal{R})$, which is $\mathcal{F} \cap (\mathcal{Q} \cup \mathcal{R})$, is stable. For ease of notation within the proof itself, assume that $\mathcal{F'}$ stands for $\mathcal{F'} \cap (\mathcal{Q} \cup \mathcal{R})$, and $\mathcal{F}$ for $\mathcal{F} \cap (\mathcal{Q} \cup \mathcal{R})$. 

Take $\mathcal{Q},\mathcal{R} \in \mathcal{P'}$. Let $\mathcal{F}_\mathcal{Q},\mathcal{F}_\mathcal{R}$ be the stable matchings over $\mathcal{Q},\mathcal{R}$, respectively.

Suppose $\mathcal{F}$ is unstable; there exists a blocking family $F \notin \mathcal{F}$. Since $F$ is blocking and $\mathcal{F'}$ is stable, $F$ cannot be an expansion of some $F'\in\mathcal{F'}$. Therefore, since $\mathcal{F}_\mathcal{Q},\mathcal{F}_\mathcal{R}$ are stable, $F=(G \cap \mathcal{Q}) \cup (H \cap \mathcal{R})$, for some $G, H \in \mathcal{F'}$.

For ease of notation, let
$$F_1=G \cap \mathcal{Q},\; F_2=G \cap \mathcal{R},\; F_3=H \cap \mathcal{Q},\;  F_4=H \cap \mathcal{R}$$
Thus, without loss of generality, $F=F_1 \cup F_4$. By definition of $\mathcal{P'}$, we have that $F_1, F_2, F_3, F_4 \in U_\mathcal{P'}$.

Consider $x \in F_1, y \in F_4$. Since $x,y \in F$ and $F$ is blocking, we have that $x \succeq_y rel_{\mathcal{F}}\big( y, prt(x)\big)$ and $y \succeq_x rel_{\mathcal{F}}\big( x, prt(y)\big)$. This means that 
$$L_x(y) \leq L_x\Big( rel_{\mathcal{F}}\big( x, prt(y)\big) \Big)$$
$$L_y(x) \leq L_y\Big( rel_{\mathcal{F}}\big( y, prt(x)\big) \Big)$$
Summing over all such $x,y$ pairs,
$$\sum_{x \in F_1 ,\: y \in F_4} L_x(y) \leq \sum_{x \in F_1 ,\: y \in F_4} L_x\Big( rel_{\mathcal{F}}\big( x, prt(y)\big) \Big)$$
$$\sum_{x \in F_1 ,\: y \in F_4} L_y(x) \leq \sum_{x \in F_1 ,\: y \in F_4} L_y\Big( rel_{\mathcal{F}}\big( y, prt(x)\big) \Big)$$
Additionally, for each $x \in F_1, y \in F_4$ there exist $z \in F_4, w \in F_1$ such that $z \succ_x rel_{\mathcal{F}}(x, prt(z))$ and $w \succ_y rel_{\mathcal{F}}(y, prt(w))$. This implies that
$$\sum_{x \in F_1 ,\: y \in F_4} L_x(y) < \sum_{x \in F_1 ,\: y \in F_4} L_x\Big( rel_{\mathcal{F}}\big( x, prt(y)\big) \Big)$$
$$\sum_{x \in F_1 ,\: y \in F_4} L_y(x) < \sum_{x \in F_1 ,\: y \in F_4} L_y\Big( rel_{\mathcal{F}}\big( y, prt(x)\big) \Big)$$
Then, by definition of $L'$, we have
$$L'_{F_1}(F_4) < L'_{F_1}(F_2)$$
$$L'_{F_4}(F_1) < L'_{F_4}(F_3)$$
But this contradicts the fact that $\mathcal{F'}$ is stable. Therefore there can be no blocking family in $\mathcal{F}$, and thus $\mathcal{F}$ is stable.
\end{proof}

\begin{thm}
All compound algorithms are stable.
\end{thm}
\begin{proof}
This statement will be proven by tracing the steps of the compound algorithm and proving that stability is preserved throughout. Take $C \in \mathcal{C}(\mathcal{P})$. 
Upon executing the algorithm, let $c$ be the final number of times that steps $1-4$ were repeated.

Nothing occurs with regards to stability in steps 1, 3, 4, and 6. In step 2, by Theorem 3 each matching $\langle\upvarepsilon(\mathcal{Q})\rangle$ is stable for all $\mathcal{Q} \in \pi$. In step 5, $\langle\upvarepsilon(\mathcal{P'})\rangle$ is stable by Theorem 3. The initial algorithm conditions and steps 2 and 5 satisfy the givens for Lemma 1, so by applying Lemma 1 $c$ times on $\langle\upvarepsilon(\mathcal{P'})\rangle$, we have that the $c^{th}$ expanded matching of $\langle\upvarepsilon(\mathcal{P'})\rangle$, $\langle C\rangle$, is stable over the original $\mathcal{P}$.
\end{proof}

\section{Open Problems}
There are several interesting questions we can now ask.
\begin{itemize}
\item determine optimal and pessimal matchings for a given party
\item determine egalitarian matchings
\item determine efficiency of the compound algorithm
\item determine additional structure of matchings under elemental and compound algorithms (dependent on chosen directed tree, etc.)
\item find applications
\end{itemize}


\begin{thebibliography}{9}
\bibitem{danilov} 
V. I. Danilov. ``Existence of stable matchings in some three-sided systems".
\textit{Mathematical Social Science} 46.2 (2003), pp. 145--148.

\bibitem{gs} 
D. Gale and L. S. Shapley. ``College admission and the stability of marriage". 
\textit{The American Mathematical Monthly} 69.1 (1962), pp.9--15.

\bibitem{gusfieldirving} 
D. Gusfield and R. W. Irving. \textit{The stable marriage problem: structure and algorithms.} Cambridge, Massachusetts: MIT Press, 1989.

\bibitem{huang} 
C. C. Huang. ``Two's company, three's a crowd: stable family and threesome roomates problems".
Technical Report TR2007-598, Computer Science Department, Dartmouth College (2007).

\bibitem{irving} 
R. W. Irving. ``Stable marriage and indifference".
\textit{Discrete Applied Mathematics} 48 (1994), pp. 261--272.

\bibitem{knuth} 
D. E. Knuth.
\textit{Mariages stables et leurs relations avec d'autre probl\`emes combinatoires.} Les Presses de l'universit\'e de Montr\'eal, 1976.

\bibitem{nghirsch} 
C. Ng and D. S. Hirschberg. ``Three-dimensional Stable Matching Problems".
\textit{SIAM Journal on Discrete Mathematics} 4.2 (1991), pp. 245--252.
\end{thebibliography}
\end{document}